\pdfoutput=1
\documentclass{sigplanconf}

\usepackage[utf8]{inputenc}
\usepackage[T1]{fontenc}
\usepackage{amsmath}
\usepackage{amsthm}
\usepackage{amssymb}
\usepackage{verbatimbox}
\usepackage{tikz}
\usepackage{bussproofs,proof}
\usepackage{multicol}
\usepackage{cancel}
\usepackage[]{algorithm2e}
\usepackage{url}

\theoremstyle{plain}

\newtheorem{theorem}{Theorem}

\theoremstyle{definition}

\theoremstyle{definition}

\newcommand{\ft}{\mathcal{T}_{FT}}
\newcommand{\ftext}{\mathcal{T}_{FT}^+}
\newcommand{\Sup}{\mathcal{S}}

\newcommand{\gtrm}{\mathcal{T}(\Sigma)}

\newcommand{\eql}{\approx}
\newcommand{\neql}{\not\approx}
\newcommand{\setof}[1]{\{#1\}}
\newcommand{\xone}[2]{#1_1,\ldots,#1_{#2}}
\newcommand{\emptyclause}{\Box}
\newcommand{\isS}{\mathbb{I}}

\newcommand{\Sub}{\mathit{Sub}}
\newcommand{\Bin}{\mathit{Bin}}
\newcommand{\leaf}{\mathit{leaf}}
\newcommand{\node}{\mathit{node}}
\renewcommand{\implies}{\rightarrow}

\begin{document}
\toappear{To appear in \textit{Proceedings of the 44th ACM SIGPLAN Symposium on Principles of Programming Languages}. ACM, 2017.}

\setlength{\pdfpageheight}{\paperheight}
\setlength{\pdfpagewidth}{\paperwidth}

\conferenceinfo{CONF 'yy}{Month d--d, 20yy, City, ST, Country}
\copyrightyear{2015}
\copyrightdata{978-1-nnnn-nnnn-n/yy/mm}
\copyrightdoi{nnnnnnn.nnnnnnn}



\title{Coming to Terms with Quantified Reasoning}

\authorinfo{Laura Kov\'acs}
          {TU Wien, Austria} {laura.kovacs@tuwien.ac.at}
\authorinfo{Simon Robillard}
          {Chalmers Univ. of Technology, Sweden}
          {simon.robillard@chalmers.se}
\authorinfo{Andrei Voronkov}
           {University of Manchester, UK\\Chalmers Univ. of Technology, Sweden}
         {andrei@voronkov.com}

\maketitle

\begin{abstract}
The theory of finite term algebras provides a natural framework to
describe the semantics of functional languages. The ability to
efficiently reason about term algebras is essential to automate
program analysis and verification for functional or imperative
programs over algebraic data types such as lists and
trees. However, as the theory of finite term algebras is not finitely
axiomatizable, reasoning about quantified properties over term
algebras is challenging. 

In this paper we address full first-order
reasoning about properties of programs manipulating term algebras, and
describe two approaches for doing so by using first-order theorem
proving.  Our first method is a conservative extension of the theory
of term algebras using a finite number of statements, while our second
method relies on extending the superposition calculus of
first-order theorem provers with additional inference rules. 

We implemented our work in the first-order theorem prover Vampire and
evaluated it on a large number of algebraic data type benchmarks, as
well as game theory constraints. Our
experimental results show that our methods are able to find proofs for many
hard problems previously unsolved by state-of-the-art methods. We also
show that Vampire implementing our methods outperforms existing SMT
solvers able to deal with algebraic data types.

\end{abstract}

\category{CR-number}{subcategory}{D 2.4 Software/Program Verification,
  F.3.1 Specifying and Verifying and Reasoning about Programs, F.3.2
  Semantics of Programming Languages, 
  F 4.1. Mathematical Logic,  I.2.3 Deduction and Theorem Proving, I.2.4 Knowledge Representation Formalisms and Methods } 



\section{Introduction}
\label{sec:intro}
Applications of program analysis and verification often require
generating and proving properties about algebraic data types, such as
lists and trees. These data types (sometimes also called recursive or
inductive data types) are special cases of term algebras, and hence
reasoning about such program properties requires proving in the
first-order theory of term algebras.  Term algebras are of particular
importance for many areas of computer science, in particular program
analysis. Terms may be used to formalize the semantics of programming
languages~\cite{goguen1977initial,clark1978negation,courcelle1983fundamental};
they can also themselves be the object of computation. The latter is
especially obvious in the case of functional programming languages,
where algebraic data structures are manipulated. Consider for example
the following declaration, in the functional language ML:
\renewcommand{\ttdefault}{pcr}
\begin{equation*}
  \texttt{\textbf{datatype} nat = zero | succ of nat;}
\end{equation*}
Although the functional programmer calls this a data type declaration,
the logician really sees the declaration of an (initial) algebra whose
signature is composed of two symbols: the constant $zero$ and the
unary function $succ$. The elements of this data type/algebra are all
ground (variable-free) terms over this signature, and programs
manipulating terms of this type can be declared by means of recursive
equations.  For example, one can define a program computing the
addition of two natural numbers by the following two equations:
\begin{equation*}
  \begin{split}
    & \texttt{add zero x = x} \\
    & \texttt{add (succ x) y = succ (add x y)}
  \end{split}
\end{equation*}
Verifying the correctness of programs manipulating this data type
usually amounts to proving the satisfiability of a (possibly
quantified) formula in the theory of this term algebra. In the case of
the program defined above, a simple property that one might want to
check is that adding a non-zero natural number to another results in a
number that is also different from zero:
\begin{equation*}
  \texttt{x $\neq$ zero $\lor$ y $\neq$ zero $\Rightarrow$ add x y $\neq$ zero}
\end{equation*}
Note that depending on the semantics of the programming language,
there may exist cyclic terms such as the one satisfying the equation
$x \approx succ(x)$, or even infinite terms, but in a strictly
evaluated language, only finite non-cyclic terms lead to terminating
programs. Since program verification is in general concerned with
program safety and termination, it is desirable to consider in
particular the theory of finite term algebras, denoted by $\ft$ in the
sequel.

The full first-order fragment of $\ft$ is known to be
decidable~\cite{malcev1962axiomatizable}. One may hence hope to easily
automate the process of reasoning about properties of programs
manipulating algebraic data types, such as lists and trees,
corresponding to term algebras. However, properties of such programs
are not confined strictly to $\ft$ for the following reasons: program
properties typically include arbitrary function and predicate symbols
used in the program, and they may also involve other theories, for
example the theory of integer/real arithmetic. Decidability of $\ft$
is however restricted to formulas that only contain term algebra
symbols, that is, uninterpreted functions, predicates and other theory
symbols cannot be used. If this is not the case, non-linear arithmetic
could trivially be encoded in $\ft$, implying thus the undecidability
of $\ft$.  Due to the decidability requirements of $\ft$ on the one
hand, and the logical structure of general program properties over
term algebras on the other hand, decision procedures based
on~\cite{malcev1962axiomatizable} for reasoning about programs
manipulating algebraic data cannot be simply used. For the purpose of
proving program properties with symbols from $\ft$, one needs more
sophisticated reasoning procedures in extensions of $\ft$.

For this purpose, the works
of~\cite{barrett2007decision,reynolds2015codatatypes} introduced
decision procedures for various fragments of the theory of term
algebras; these techniques are implemented as satisfiability modulo
theory (SMT) procedures, in particular in the CVC4 SMT
solver~\cite{CVC4}. However, these results target mostly reasoning in
quantifier-free fragments of term algebras.  To address this challenge
and provide efficient reasoning techniques with both quantifiers and
term algebra symbols, in this paper we propose to use first-order
theorem provers. We describe various extensions of the superposition
calculus used by first-order theorem provers and adapt the saturation
algorithm of theorem provers used for proof search.

Theory-specific reasoning in saturation-based theorem provers is
typically conducted by including the theory axioms in the set of input
formulas to be saturated. Unfortunately a complete axiomatization of
the theory of term algebras requires an infinite number of sentences:
the \emph{acyclicity rule}, which ensures that a model does not
include cyclic terms, is described by an infinite number of
inequalities $x \not\approx f(x)$, $x \not\approx f(f(x)),
\dots$ This property of term algebras prevents us from performing
theory reasoning in saturation-based proving in the usual
way.

As a first attempt to remedy this state of affairs, in this paper we
present a conservative extension of the theory of term algebras that
uses a finite number of sentences (Section~\ref{sec:extension}). This
extension relies on the addition of a predicate to describe the
``proper subterm'' relation between terms. This approach is complete
and can easily be used in any first-order theorem prover without any
modification.

Unfortunately, the subterm relation is transitive, so that the number
of predicates produced by saturation quickly becomes a burden for any
prover. To improve the efficiency of the reasoning, we offer an
alternative solution: extending the inference system of the saturation
theorem prover with additional rules to treat equalities between terms
(Section~\ref{sec:calculus}).

We implemented our new inference system, as well as the subterm
relation, in the first-order theorem prover
Vampire~\cite{Vampire13}. We tested our implementation on two sets of
benchmarks. We used 4170 problems describing properties of functional
programs manipulating algebraic data types; these problems were taken
from~\cite{reynolds2015codatatypes}.  This set of examples were
generated using the Isabelle inductive theorem prover~\cite{Isabelle}
and translated by the Sledgehammer
system~\cite{blanchette2013sledgehammer}.  Further, we also used
problems from~\cite{colmerauer2000expressiveness} with many quantifier
alternations over term algebras. When compared to state-of-the-art SMT
solvers, such as CVC4 and Z3~\cite{Z3}, our experimental results give
practical evidence of the efficiency and logical strength of our work:
many hard problems that could not be solved before by any existing
technique can now be solved by our work (see
Section~\ref{sec:experiments}).

\noindent{\bf Contributions.} The main contributions of our paper are
summarized below.

\begin{itemize}
\item We extend the theory $\ft$ of finite term algebras with a
  subterm relation denoting proper subterm relations between terms. We
  call this extension $\ftext$ and prove that $\ftext$ is a
  conservative extension of $\ft$. When compared to $\ft$, the
  advantage of $\ftext$ is that it is finitely axiomatizable and hence
  can be used by any first-order theorem prover. Moreover, one can
  combine $\ftext$ with other theories, going even to undecidable
  fragments of the combined theory of term algebras and other
  theories. As an important consequence of this conservative
  extension, our work yields a superposition-based decision procedure
  for term algebras (Section~\ref{sec:extension}).

\item We show how to optimize superposition-based first-order
  reasoning using new, term algebra specific, simplification rules,
  and an incomplete, but simple, replacement for a troublesome
  acyclicity axiom. Our new inference system provides an alternative
  and efficient approach to axiomatic reasoning about term algebras in
  first-order theorem proving and can be used with combinations of
  theories (Section~\ref{sec:calculus}).

\item We implement our work in the first-order theorem prover
  Vampire. Our works turns Vampire into the first first-order theorem
  prover able to reason about term algebras, and therefore about
  algebraic data types. Our experiments show that our implementation
  outperforms state-of-the-art SMT solvers able to reason with
  algebraic data types. For example, Vampire solved 50 SMTLIB problems
  that could not be solved by any other solver before
  (Section~\ref{sec:experiments}).
\end{itemize}

\section{Preliminaries}
\label{sec:prelim}
We consider standard first-order predicate logic with equality. The
equality symbol is denoted by $\eql$.  We allow all standard boolean
connectives and quantifiers in the language. We assume that the
language contains the logical constants $\top$ for always true and
$\bot$ for always false formulas.

Throughout this paper, we denote terms by $r, s, u, t$, variables by
$x, y, z$, constants by $a, b, c, d$, function symbols by $f, g$ and
predicate symbols by $p, q$, all possibly with indices. We consider
equality $\eql$ as part of the language, that is, equality is not a
symbol. For simplicity, we write {$s \neql t$} for the formula
$\neg(s \eql t)$.

An \emph{atom} is an equality or a formula of the form $p(\xone{t}{n})$,
where $p$ is a predicate symbol and $\xone{t}{n}$ are terms. A \emph{literal} is
an atom $A$ or its negation $\neg A$. Literals that are atoms are called
\emph{positive}, while literals of the form $\neg A$ are \emph{negative}. A
\emph{clause} is a disjunction of literals $L_1 \vee \ldots \vee L_n$, where $n
\geq 0$. When $n=0$, we will speak of the {empty clause}, denoted by
{$\emptyclause$}. The empty clause is always false.

We denote atoms by $A$, literals by $L$, clauses by $C,D$, and
formulas by $F,G$, possibly with indices.

Let $F$ be a formula with free variables $\bar{x}$, then $\forall F$
(respectively, $\exists F$) denotes the formula $(\forall \bar{x}) F$
(respectively, $(\exists \bar{x}) F$). A formula is
called \emph{closed}, or a \emph{sentence}, if it has no free
variables. A formula or a term is called \emph{ground} if it has no
occurrences of variables.

A \emph{signature} is any finite set of symbols.  The \emph{signature
of a formula} $F$ is the set of all symbols occurring in this
formula. For example, the signature of $(\forall x) b \eql g(x)$ is
$\setof{g,b}$. When we speak about a \emph{theory}, we either mean a
set of all logical consequences of a set of formulas
(called \emph{axioms} of this theory), or a set of all formulas valid
on a class of first-order structures. Specifically, we are interested
in the theories of term algebras, in which case we use the second
meaning. When we discuss a theory, we call symbols occurring in the
signature of the theory \emph{interpreted}, and all other
symbols \emph{uninterpreted}.

By an \emph{expression $E$} we mean a term, atom, literal, or clause. A
\emph{substitution} $\theta$ is a finite mapping from variables to terms.
An \emph{application} of this substitution to an expression (e.g. a
term or a clause) $E$, denoted by {$E\theta$}, is the expression
obtained from $E$ by the simultaneous replacement of each variable $x$
in it, such that $\theta(x)$ is defined, by $\theta(x)$.  We write
$E[s]$ to mean an expression $E$ with a particular occurrence of a
term $s$.  A \emph{unifier} of two expressions $E_1$ and $E_2$ is a
substitution $\theta$ such that $E_1\theta=E_2\theta$. It is known
that if two expressions have a unifier, then they have a so-called
\emph{most general unifier (mgu)} -- see \cite{Robinson65} for details on
computing mgus.

\section{Superposition and Proof Search}
\label{sec:inference}
We now recall some terminology related to inference systems and
first-order theorem proving.  Inference systems are used in the theory
of
superposition~\cite{NieuwenhuisRubio:HandbookAR:paramodulation:2001}
implemented by several leading automated first-order theorem provers,
including Vampire~\cite{Vampire13} and E~\cite{E:Schulz:2002}.  The
material of this section is based on~\cite{Vampire13}, adapted to our
setting.

\subsection{The Superposition Inference System}

\begin{figure*}
  \begin{itemize}
  \item Resolution
    \[
    \AxiomC{$\underline{A} \lor C_1$}
    \AxiomC{$\underline{\lnot A'} \lor C_2$}
    \BinaryInfC{$(C_1\lor C_2)\sigma$}
    \DisplayProof
    \hspace{6em}
    \AxiomC{$\underline{s \not\approx s'} \lor C$}
    \UnaryInfC{$C \theta$}
    \DisplayProof
    \]
    where $A$ is not an equality predicate, $\sigma = mgu(A, A')$ and $\theta = mgu(s, s')$

  \item Superposition
    \[
    \AxiomC{$\underline{l \approx r} \lor C_1$}
    \AxiomC{$\underline{L \lbrack l' \rbrack} \lor C_2$}
    \BinaryInfC{$(C_1 \lor L \lbrack r \rbrack \lor C_2) \theta$}
    \DisplayProof
    \hspace{3em}
    \AxiomC{$\underline{l \approx r} \lor C_1$}
    \AxiomC{$\underline{t \lbrack l' \rbrack \approx t'} \lor C_2$}
    \BinaryInfC{$(C_1 \lor t \lbrack r \rbrack \approx t' \lor C_2) \theta$}
    \DisplayProof
    \hspace{3em}
    \AxiomC{$\underline{l \approx r} \lor C_1$}
    \AxiomC{$\underline{t \lbrack l' \rbrack \not\approx t'} \lor C_2$}
    \BinaryInfC{$(C_1 \lor t \lbrack r \rbrack \not\approx t' \lor C_2) \theta$}
    \DisplayProof
    \]
    where $l'$ not a variable, $L$ is not an equality,
    $\theta = mgu(l, l')$, $l \theta \npreceq r \theta$ and $t \lbrack
    l' \rbrack \theta \npreceq t'\theta$

  \item Factoring
    \[
    \AxiomC{$\underline{A} \lor \underline{A'} \lor C$}
    \UnaryInfC{$(A \lor C)\sigma$}
    \DisplayProof
    \hspace{6em}
    \AxiomC{$s \approx t \lor \underline{s' \approx t'} \lor C$}
    \UnaryInfC{$(s \approx t \lor t \not\approx t' \lor C)\theta$}
    \DisplayProof
    \]
    where $\sigma = mgu(A, A')$, $\theta = mgu(s, s')$, $s\theta
    \npreceq t\theta$ and $t\theta \npreceq t'\theta$
  \end{itemize}
  \caption{The superposition calculus $\Sup$\label{fig:sup}.}
\end{figure*}

First-order theorem provers perform inferences
on clauses. 
An \emph{inference rule} is an $n$-ary relation on
formulas, where $n \geq 0$. The elements of such a relation are called
\emph{inferences} and usually written as:
\[
\infer[.]{
  C
}{
  C_1 & \ldots & C_{n}
}
\]
The clauses $\xone{C}{n}$ are called the \emph{premises} of this
inference, whereas the clause $C$ is the \emph{conclusion} of the
inference.  An \emph{inference system} $\isS$ is a set of inference
rules. An \emph{axiom} of an inference system is any conclusion of an
inference with $0$ premises.  Any inferences with $0$ premises and a
conclusion $C$ will be written without the bar line, simply as $C$.

Modern first-order theorem provers use and implement the
\emph{superposition inference system}, which is parametrized by a
\emph{simplification ordering} $\preceq$ on terms and a
\emph{selection function} on clauses.

An ordering $\preceq$ on terms is called a \emph{simplification
  ordering} if it satisfies the following conditions:
\begin{enumerate}
\item $\preceq$ is \emph{well-founded}: there exists no infinite
  sequence of terms $t_0, t_1, \ldots$ such that $t_0\preceq
  t_1\preceq\ldots$;
\item $\preceq$ is \emph{stable under substitution}: if $s \preceq t$
  then $s\theta \prec t\theta$, for every term $s,t$ and substitution
  $\theta$;
\item $\preceq$ is \emph{monotonic}: if $s \preceq t$ then $l \lbrack
  s \rbrack \preceq l \lbrack t \rbrack $ for all terms $l, s, t$;
\item $\preceq$ has the \emph{subterm property}: if $s$ is a proper
  subterm of $t$, then $s \preceq t$.
\end{enumerate}
Given two terms $s\preceq t$, we say that $s$ is smaller than $t$ and
$t$ is larger/greater than $s$ wrt $\preceq$. This ordering $\preceq$
can be extended to literals and clauses.

A \emph{selection function} selects in every non-empty clause a
non-empty subset of literals. In the following, we underline literals
to indicate that they are selected in a clause; that is we write
$\underline{L} \lor C$ to denote that the literal $L$ is selected. A
selection function is said to be \emph{well-behaved} if in a given
clause it selects either a negative literal or all the maximal
literals wrt the simplification ordering $\preceq$.
 
We now fix a simplification ordering $\preceq$ and a well-behaved
selection function and define the \emph{superposition inference
  system}. This inference system, denoted by $\Sup$, consists of the
inference rules of Figure~\ref{fig:sup}. The inference system $\Sup$
is a sound and refutationally complete inference system for
first-order logic with equality. By refutational completeness we mean
that if a set $S$ of formulas is unsatisfiable, then $\emptyclause$ is
derivable from $S$ in $\Sup$.

\subsection{Proof Search by Saturation}

Superposition theorem provers implement proof-search algorithms in
$\Sup$ using so-called \emph{saturation algorithms}, as follows. Given
a set $S$ of formulas, superposition-based theorem provers try to
saturate $S$ with respect to $\Sup$, that is build a set of formulas
that contains $S$ and is closed under inferences in $\Sup$.  At every
step, a saturation algorithm selects an inference of $\Sup$, applies
this inference to $S$, and adds conclusions of the inferences to the
set $S$. If at some moment the empty clause $\emptyclause$ is
obtained, by soundness of $\Sup$, we can conclude that the input set
of clauses is unsatisfiable.  To ensure that a saturation algorithm
preserves completeness of $\Sup$, the inference selection strategy
must be {fair}: every possible inference must be selected at some step
of the algorithm. A saturation algorithm with a fair inference
selection strategy is called a \emph{fair saturation algorithm}.

A naive implementation of fair saturation algorithms based on $\Sup$
will not yield however an efficient theorem prover. This is because at
every step of the saturation algorithm, the number of clauses in the
set $S$ of clauses, representing the proof-search space,
grows. Therefore, for the efficiency of organizing proof search, one
needs to use the notion of \emph{redundancy}, which allows to delete
so-called {redundant clauses} during saturation from the search space.
A clause $C\in S$ is \emph{redundant} in $S$ if it is a logical
consequence of those clauses in $S$ that are strictly smaller than $C$
w.r.t. the simplification ordering $\preceq$. In a nutshell,
saturation algorithms using redundancy not only generate but also
delete clauses from the set $S$ of clauses.  Deletion of redundant
clauses is desirable since every deletion reduces the search space.
If a newly generated clause $C'$ during one step of saturation makes
some clauses in $S$ redundant, adding $C'$ to the search space will
remove other (more complex) clauses from $S$. This observation is
exploited by first-order theorem provers in the process of
prioritizing inferences during inference selection, giving rise to
so-called \emph{simplifying} and \emph{generating} inferences.
Simplifying inferences make one or more clauses in the search space
redundant and thus delete clauses from the search space. That is, an
inference
\[
\infer[.]{ C }{ C_1 & \ldots & C_{n} }
\]
is called \emph{simplifying} if at least one of the premises $C_i$
becomes redundant (and deleted) after the addition of the conclusion
$C$ to the search space.  Inferences that are not simplifying are
\emph{generating}: instead of simplifying clauses in the search space,
they generate and add a new clause to the search space.  Efficient
saturation algorithms exploit simplifying and generating inferences,
as follows: from time to time provers try to search for simplifying
inferences at the expense of delaying generating inferences.

\section{The Theory of Finite Term Algebras}
\label{sec:theory}
A definition of the first-order theory of term algebras over a finite
signature can be found in
e.g.~\cite{RybinaVoronkov:termalgebras:2001}, along with an
axiomatization of this theory and a proof of its completeness. In this
section we overview this theory and known results about it.

\subsection{Definition}

Let $\Sigma$ be a finite set of function symbols containing at least
one constant. Denote by $\gtrm$ the set of all ground terms built from
the symbols in $\Sigma$.

The \emph{$\Sigma$-term algebra} is the algebraic structure whose
carrier set is $\gtrm$ and defined in such a way that every ground
term is interpreted by itself (we leave details to the reader). We
will sometimes consider extensions of term algebras by additional
symbols. Elements of $\Sigma$ will be called \emph{term constructors}
(or simply just \emph{constructors}), to distinguish them from other
function symbols. The $\Sigma$-term algebra will also be denoted by
$\gtrm$.

Consider the following set of formulas.

\begin{equation}\label{eqn:exhaust}
  \bigvee_{f \in \Sigma} \exists \overline{y} \; (x \eql f(\overline{y})) \tag{A1}
\end{equation}

\begin{equation}\label{eqn:dist}
  f(\overline{x}) \neql g(\overline{y}) \tag{A2}
\end{equation}
for every $f, g \in \Sigma$ such that $f \neq g$;

\begin{equation}\label{eqn:inj}
  f(\overline{x}) \eql f(\overline{y}) \rightarrow \overline{x} \eql \overline{y} \tag{A3}
\end{equation}
for every $f \in \Sigma$ of arity $\geq 1$;

\begin{equation}\label{eqn:acyclic}
  t \neql x \tag{A4}
\end{equation}
for every non-variable term $t$ in which $x$ appears.

Some of these formulas contain free variables, we assume that they are
implicitly universally quantified.

Axiom (\ref{eqn:exhaust}), sometimes called the domain closure axiom,
asserts that every element in $\Sigma$ is obtained by applying a term
constructor to other elements.

Axiom (\ref{eqn:inj}) describes the injectivity of term constructors,
while axiom (\ref{eqn:dist}) expresses the fact that terms constructed
from different constructors are distinct. Throughout this paper, we
refer to (\ref{eqn:dist}) as the distinctness axiom and to
(\ref{eqn:inj}) as the injectivity axiom.

The axiom schema (\ref{eqn:acyclic}), called the acyclicity axiom,
asserts that no term is equal to its proper subterm, or in other words
that there exist no cyclic terms.

In the following sections we will also discuss theories in which there
are non-constructor function symbols. Note that when we deal with such
theories, the acyclicity axioms are used only when all symbols in $t$
are constructors.

\subsection{Known Results}

We denote by $\ft$ the theory axiomatized by
(\ref{eqn:exhaust})--(\ref{eqn:acyclic}), that is, the set of logical
consequences of all formulas in
(\ref{eqn:exhaust})--(\ref{eqn:acyclic}). Note that the $\Sigma$-term
algebra is a model of all formulas
(\ref{eqn:exhaust})--(\ref{eqn:acyclic}), and therefore also a model
of $\ft$.

\begin{theorem}\label{thm:completeness}%
  The following results hold.
  \begin{enumerate}
  \item $\ft$ is complete. That is, for every sentence $F$ in the
    language of $\gtrm$, either $F \in \ft$ or $(\neg F) \in \ft$.
  \item $\ft$ is decidable.
  \item If $\Sigma$ contains at least one symbol of arity $> 1$, then
    the first-order theory of $\ft$ is non-elementary.
  \end{enumerate}
\end{theorem}
Completeness of $\ft$ is proved in a number of papers - a detailed
proof can be found in, e.g., \cite{RybinaVoronkov:termalgebras:2001}.

Decidability of $\ft$ in Theorem~\ref{thm:completeness} is implied by
the completeness of $\ft$ and by the fact that $\ft$ has a recursive
axiomatization. More precisely, completeness gives the following
(slightly unusual) decision procedure: given a sentence $F$, run any
complete first-order theorem proving procedure (e.g., a complete
superposition theorem prover) simultaneously and separately on $F$ and
$\neg F$. We can get around the problem that the axiomatisation is
infinite but throwing in axioms, one after one, while running the
proof search --- indeed, by the compactness property of first-order
logic, if a formula $G$ is implied by an infinite set of formulas, it
is also implied by a finite subset of this set. One of contributions
of this paper is showing how to avoid dealing with infinite
axiomatizations.

Further, the non-elementary property of $\ft$ in
Theorem~\ref{thm:completeness} follows from a result
in~\cite{FerranteRackoff}: every theory in which one can express a
pairing function has a hereditarily non-elementary first-order theory.

Note that the completeness of $\ft$ implies that $\ft$ is exactly the
set of all formulas true in the $\Sigma$-term algebra. First-order
theories of term algebras are closely related to non-recursive logic
programs, for related complexity results, also including the case with
only unary functions, see~\cite{VoVo98}.

Let us make the following important observation. The decidability and
other results of Theorem~\ref{thm:completeness} do not hold when
uninterpreted functions or predicates are added to $\ft$. If we add to
the $\Sigma$-term algebra uninterpreted symbols, one can for example
use these symbols to provide recursive definitions of addition and
multiplication, thus encoding first-order Peano arithmetic.  Using the
same reasoning as in~\cite{KorovinVoronkov:LASCA:2007} one can then
prove the following result.

\begin{theorem}\label{thm:pi11}
  The first-order theory of $\Sigma$-algebras with uninterpreted
  symbols is $\Pi_1^1$-complete, when $\Sigma$ contains at least one
  non-constant.
\end{theorem}

We will not give a full proof of Theorem~\ref{thm:pi11} but refer
to~\cite{KorovinVoronkov:LASCA:2007} for details. Here, we only show
how to encode non-linear arithmetic in $\ft$ using $\Sigma$-term
algebra uninterpreted symbol, which is relatively
straightforward. Assume, without loss of generality, that $\Sigma$
contains a constant $0$ and a unary function symbol $s$
(successor). Then all ground terms, and hence all term algebra
elements are of the form $s^n(0)$, where $n \geq 0$. We will identify
any such term $s^n(0)$ with the non-negative integer $n$.

Add two uninterpreted functions $+$ and $\cdot$ and consider the set
$A$ of formulas defined as follows:
  
\[
\forall x \; (x + 0 = x)
\]

\[
\forall x \forall y \; (s(x) + y = s(x+y))
\]

\[
\forall x \; (x \cdot 0 = 0)
\]

\[
\forall x \forall y \; (s(x) \cdot y = (x \cdot y) + y)
\]

It is not hard to argue that in any extension of the $\Sigma$-algebra
satisfying $A$, the functions $+$ and $\cdot$ are interpreted as the
addition and multiplication on non-negative integers. Let now $G$ be
any sentence using only $+,\cdot,s,0$. Then we have that $A \implies
G$ is valid in the $\Sigma$-algebra if and only if $G$ is a true
formula of arithmetic.

Note that Theorem~\ref{thm:pi11} refers to the theory of algebras,
i.e. the set of formulas valid on $\Sigma$-algebra. In view of this
theorem, with uninterpreted symbols of arity $\geq 1$ in the
signature, this includes more formulas than the set of formulas
derivable from (\ref{eqn:exhaust})--(\ref{eqn:acyclic}).

\subsection{Other Formalizations}

Instead of using existential quantifiers in (\ref{eqn:exhaust}), one
can also use axioms based on destructors (or projection functions) of
the algebra. For all function symbols $f$ of arity $n > 0$ and all $i
= 1,\ldots,n$, introduce a function $p_f^i$. The \emph{destructor
  axioms} using these functions are:

\begin{equation}\label{eqn:exhaust2}
  x \eql f(p_f^1(x), \ldots, p_f^n(x)). \tag{A1'}
\end{equation}

The axiom (\ref{eqn:inj}) can be replaced by the following axioms,
which can be considered as a definition of destructors:
\begin{equation}\label{eqn:proj}
  p_f^i(f(x_1, \ldots,x_i, \ldots, x_n)) \eql x_i \tag{A3'}
\end{equation}
Given the other axioms, (\ref{eqn:inj}) and (\ref{eqn:proj}) are
logically equivalent, but some authors prefer the presentation based
on destructors. Note, however, that the behavior of a destructors
$p_f^i$ is unspecified on some terms.

\subsection{Extension to Many-Sorted Logic}

In practice, it can be useful to consider multiple sorts, especially
for problems taken from functional programming. In this setting, each
term algebra constructor has a type $\tau_1 \times \dots \times \tau_n
\to \tau$. The requirement that there is at least one constant should
then be replaced by the requirement that for every sort, there exists
a ground term of this sort.

We can also consider similar theories, which mix constructor and
non-constructor sorts. That is, some sorts contain constructors and
some do not.

Consider an example with the following term algebra signature:
\[
\Sigma_{\Bin} = \{ \leaf : \tau \to \Bin, \node : \Bin \times \tau \times \Bin \to \Bin \}
\]
This signature defines an algebra of binary trees, where every node
and leaf is decorated by an element of a (non-constructor) sort
$\tau$. In this case term algebra axioms are only using sorts with
constructors. The axioms of this theory of trees, as defined
previously, are shown in Figure~\ref{fig:tree-axioms}.

\begin{figure*}
  \[
  \exists y \big(x \eql \leaf(y)\big) \lor \exists y_1, y_2, y_3 \big(x \eql \node(y_1, y_2, y_3)\big)
  \]

  \[
  \node(x_1, x_2, x_3) \neql \leaf(y_1)
  \]

  \[
  \leaf(x) \eql \leaf(y) \rightarrow x \eql y
  \]

  \[
  \node(x_1, x_2, x_3) \eql \node(y_1, y_2, y_3) \rightarrow x_1 \eql y_1 \land x_2 \eql y_2 \land x_3 \eql y_3
  \]

  \[
  x \neql \node(x, y_1, y_2) \qquad x \neql \node(y_1, y_2, x) \qquad x \neql \node(\node(x, y_1, y_2), y_3, y_4) \qquad \dots
  \]
  \caption{The instantiation of the theory axioms for the signature
    $\Sigma_{\Bin}$\label{fig:tree-axioms}.}
\end{figure*}

\section{A Conservative Extension of the Theory of Term Algebras}
\label{sec:extension}
In this paper we aim to prove theorems in first-order theories
containing constructor-defined types. While in general the theory is
$\Pi_1^1$-complete, we still want to have a method that behaves well
in practice. Our method will be based on extending the superposition
calculus by axioms and/or rules for dealing with term algebra
constructor symbols.

One of the criteria of behaving well in practice is to have a method
that is complete for pure term algebra formulas, that is, without
uninterpreted functions. The immediate idea would be to use the
axiomatization of term algebras consisting of
(\ref{eqn:exhaust})--(\ref{eqn:acyclic}), however this does not work
since there is an infinite number of acyclicity axioms.

In this section we show how to overcome this problem by using an
extension of term algebras by a binary relation $\Sub$, denoting the
proper subterm relation. Let us further denote by $\ftext$ the set of
formulas which contains~\eqref{eqn:exhaust}--\eqref{eqn:inj}, but
replaces the acyclicity axiom~\eqref{eqn:acyclic} by the following
axioms \eqref{eqn:proper}--\eqref{eqn:subacy}:

\begin{equation}\label{eqn:proper}
  \Sub(x_i, f(x_1, \dots, x_i, \dots, x_n)), \tag{B1}
\end{equation}
for every $f \in \Sigma$ of arity $n \geq 1$ and every $i$ such that $n \geq i
\geq 1$.

\begin{equation}\label{eqn:trans}
  \Sub(x,y) \land \Sub(y,z) \rightarrow \Sub(x,z) \tag{B2}
\end{equation}

\begin{equation}\label{eqn:subacy}
  \lnot \Sub(x,x) \tag{B3}
\end{equation}

Intuitively, the predicate $\Sub(s,t)$ holds iff $s$ is a proper
subterm of $t$. Axiom (\ref{eqn:proper}) ensures that this relation
holds for terms $s$ appearing directly under a term algebra
constructor in $t$ , while (\ref{eqn:trans}) describes the
transitivity of the subterm relation and ensures that the relation
also holds if $s$ is more deeply nested in $t$. Axiom
(\ref{eqn:subacy}) asserts that no term may be equal to its own proper
subterm.

We now  observe the following properties of
\eqref{eqn:proper}--\eqref{eqn:subacy}. 

\begin{theorem}\label{thm:conservative}%
  $\ftext$ is a conservative extension of $\ft$, that is:
  \begin{enumerate}
  \item \label{itm:noloss} Every theorem in $\ft$ is a theorem in $\ftext$;
  \item \label{itm:nogain} Every theorem in $\ftext$ that uses only
    symbols from the language of $\ft$ (i.e. not using the predicate
    $\Sub$) is also a theorem of $\ft$.
  \end{enumerate}
\end{theorem}
\begin{proof}
For \eqref{itm:noloss}, it is enough to prove that every instance of
the acyclicity axiom~\eqref{eqn:acyclic} of $\ft$ is implied by axioms
of $\ftext$. To this end, note that for every term $t$ and its proper
subterm $s$, \eqref{eqn:proper}--\eqref{eqn:trans} imply $\Sub(s,t)$,
so every instance of the acyclicity axiom~\eqref{eqn:acyclic} is
implied by \eqref{eqn:proper}--\eqref{eqn:subacy}.

To prove part \eqref{itm:nogain}, first note that $\ftext$ is
consistent (sound). This follows from the fact that it has a model,
which extends the $\Sigma$-term algebra by interpreting $\Sub$ as the
subterm relation. Now assume, by contradiction, that there is a
sentence $F$ not using $\Sub$ such that $F \in
\ftext$ and $F \not\in \ft$. By the completeness result of
Theorem~\ref{thm:completeness}, we then have $\neg F \in \ft$, which by part
\eqref{itm:noloss} implies $\neg F \in \ftext$. We have both $F \in \ftext$ and
$\neg F \in \ftext$, which contradicts the consistency of $\ftext$.
\end{proof}

Note that the full first-order theory of term algebras with the
subterm predicate is undecidable \cite{venkataraman1987decidability}.

The important difference between $\ft$ and $\ftext$ is
that \emph{$\ftext$ is finitely axiomatizable}. This fact and
Theorem~\ref{thm:conservative} can be directly used to
design \emph{superposition-based proof procedures} for $\ft$, as
follows.  Given a term algebra sentence $F$, we can search for a
superposition proof of $F$ from the axioms of $\ftext$. Such a proof
exists if and only if $F$ holds in the $\Sigma$-term algebra. This
proof procedure can even be turned into a
\emph{superposition-based decision
procedure for $\ft$}, which is based on attempting to prove $F$ and
$\neg F$ in parallel, until one of them is proved, which is guaranteed
by the completeness of $\ft$ from Theorem~\ref{thm:completeness}.

It is interesting that, while proving a formula $F$ with quantifier
alternations in this way, first-order theorem provers will first
skolemize $F$, introducing uninterpreted functions. While the
first-order theory of term algebras with arbitrary uninterpreted
functions is incomplete, our results guarantee \emph{completeness on
formulas with uninterpreted functions obtained by skolemization}. This
is so because skolemization preserves validity and hence, using
Theorem~\ref{thm:conservative}, we conclude completeness on skolemized
formulas with uninterpreted functions.

While it is hard to expect that proving term algebra formulas by
superposition will result in a better decision procedure compared to
those described in the literature, see
e.g.~\cite{colmerauer2000expressiveness}, our approach has the
advantage that it can be combined with other theories and can be used
for proving formulas in undecidable fragments of the full first-order
theory of term algebras. Given a formula containing both constructors,
uninterpreted symbols and possible theory symbols, we can attempt to
prove this formula by adding the axioms of $\ftext$ and then use a
superposition theorem prover.  The results of this section show that
this method is strong enough to prove all (pure) term algebra
theorems. Our experimental results described in
Section~\ref{sec:experiments} give an evidence that it is also
efficient in practice.

The conservative extension $\ftext$ presented above thus allows one to
encode problems in the theory of term algebras and reason about them
using any tool for automated reasoning in first-order logic. However
the transitive nature of the predicate $\Sub$ can impact the
performance of provers negatively. Note that the transitivity axiom
can also be replaced by axioms of the form:

\[
  \Sub(x,x_i) \rightarrow \Sub(x, f(x_1, \dots, x_i, \dots, x_n)).
\]
Using these new axioms will result in fewer inferences during proof
search and a slower growth of the subterm relation, which are
important parameters for the provers' performance.

\section{An Extended Calculus}
\label{sec:calculus}
In this section we describe an alternative way to use superposition
theorem provers to reason about term algebras. Instead of including
theory axioms in the initial set of clauses, we extend the calculus
with inferences rules. This is similar to the way paramodulation is
used to replace the axiomatization of equality, apart from the fact
that we cannot obtain a calculus that is complete.

\subsection{A naive calculus}

In this section we will consider alternatives and improvements to
axiomatizing term algebras. The idea is to add simplification rules
specific to term algebras and replace the troublesome acyclicity axiom
by special purpose inference rules.

The superposition calculus uses term and clause orderings to orient
equalities, restrict the number of possible inferences, and
simplification. The general rule is that a clause in the search space
can be deleted if it is implied by strictly smaller clauses in the
search space.

One obvious idea is to add several simplification rules, corresponding
to applications of resolution and/or superposition to term algebra
axioms.  For example, a clause $f(s) \eql s \lor C$ can be replaced by
a simpler, yet equivalent, clause $C$.  Likewise, the clause $f(s)
\eql f(t) \lor C$ is equivalent, by injectivity of the constructors,
to the clause $s \eql t \lor S$. The clause $s \eql t \lor S$ is also
smaller than $f(s) \eql f(t) \lor C$, so it can replace this clause.

Let us start with examples showing that replacing axioms
by rules can result in incompleteness even in very simple cases.

Take for example two ground unit clauses $f(a) \eql b$ and $g(a) \eql
b$, where all symbols apart from $b$ are constructors. This set of
clauses is unsatisfiable in the theory of term algebras. However, if
we replace the axiom $f(x) \neql g(y)$ by a simplification rule, there
are no inferences that can be done between these clauses (assuming we
are using the standard Knuth-Bendix ordering).

Another example showing that the acyclicity axiom can be hard to drop
or replace is the set of two ground unit clauses $f(a) \eql b$ and
$f(b) \eql a$, where $f$ is a constructor. This set of clauses is also
unsatisfiable in the theory of term algebras, since it implies
$f(f(b))=b$. Similar to the previous example, there is no
superposition inference between these two clauses.

\subsection{The Distinctness Rule}

We implemented an extra simplification and a deletion rule. Such rules
will be denoted using a double line, meaning that the clauses in the
premise are replaced by the clauses in the conclusion.

The simplification rule is 

\begin{prooftree}
  \AxiomC{$f(s) \eql g(t) \lor A$}
  \doubleLine\
  \RightLabel{{\scriptsize Dist-S$^+$},}
  \UnaryInfC{$A$}
\end{prooftree}
where $f$ and $g$ are different constructors. Essentially, it removes
from the clause a literal false in the theory of term algebras.

The deletion rule is

\begin{prooftree}
  \AxiomC{$f(s) \neql g(t) \lor A$}
  \doubleLine\
  \RightLabel{{\scriptsize Dist-S$^-$},}
  \UnaryInfC{$\emptyset$}
\end{prooftree}
where $f$ and $g$ are different constructors. It deletes a theory
tautology.

\subsection{The Injectivity Rule}

There is a simplification rule based on the injectivity axiom
\eqref{eqn:inj}. Suppose that $f$ is a constructor of arity $n >
0$. Then we can use the simplification rule

\begin{prooftree}
  \AxiomC{$f(s_1 \dots s_n) \eql f(t_1,\ldots,t_n) \lor C$}
  \doubleLine
  \RightLabel{.}
  \UnaryInfC{$\begin{array}{c}
                s_1 \eql t_1 \lor C \\
                \cdots \\
                s_n \eql t_n \lor C
              \end{array}$}
\end{prooftree}

One can also note that under some additional restrictions the
following inference

\begin{prooftree}
  \AxiomC{$f(s_1 \dots s_n) \neql f(t_1,\ldots,t_n) \lor C$}
  \doubleLine\
  \UnaryInfC{$s_1 \neql t_1 \lor \ldots \lor s_n \neql t_n \lor C$}
\end{prooftree}
can be considered as a simplification rule too. The restriction is the
clause ordering condition $\{s_1 \neql t_1 \lor \ldots \lor s_n \neql
t_n\} \prec C$.

Note that in both rules the premise is logically equivalent to the
conjunction of the formulas in the conclusion in the theory of term
algebras and all formulas in the conclusion are smaller than the
formula in the premise (subject to the ordering condition for the
second rule).

\subsection{The Acyclicity Rule}

Similar to the distinctness axiom and rules, we can introduce a
simplification and a deletion rule based on the acyclicity
axiom. First, we introduce a notion of a \emph{constructor subterm} as
the smallest transitive relation that each of the terms $t_i$ is a
constructor subterm of $f(t_1,\ldots,t_n)$, where $f$ is a constructor
and $n \geq i \geq 1$. For example, if $f$ is a binary constructor,
and $g$ is not a constructor, then all constructor subterms of the
term $f(f(x,a),g(y))$ are $f(x,a)$, $x$, $a$ and $g(y)$. Its subterm
$y$ is not a constructor subterm. One can easily show that any
inequality $s \neql t$, where $s$ is a constructor subterm of $t$ is
false in any extension of term algebras.

The simplification rule for acyclicity is

\begin{prooftree}
  \AxiomC{$s \eql t \lor A$}
  \doubleLine
  \RightLabel{,}
  \UnaryInfC{$A$}
\end{prooftree}
where $s$ is a constructor subterm of $t$. It deletes from a clause
its literal false in all term algebras.

The deletion rule is

\begin{prooftree}
  \AxiomC{$s \neql t \lor A$}
  \doubleLine
  \RightLabel{,}
  \UnaryInfC{$\emptyset$}
\end{prooftree}
where $s$ is a constructor subterm of $t$.
It deletes a theory tautology.

Further, if we wish to get rid of the subterm relation $\Sub$, we can
use various rules to treat special cases of acyclicity. If we do this,
we will lose completeness even for pure term algebra formulas, but
such a replacement can deal with some formulas more efficiently, while
still covering a sufficiently large set of problems.

One example of such a special acyclicity rule is the following:

\begin{prooftree}
  \AxiomC{$t \eql u \lor A$}
  \UnaryInfC{$s \neql u \lor A$}
\end{prooftree}
where $s$ is a constructor subterm of $t$. Note that this rule is not
a simplification rule, so we do not delete the premise after applying
this rule.

\section{Experimental Results}
\label{sec:experiments}
\subsection{Implementation}

We implemented the subterm relation of Section~\ref{sec:extension} and
simplification rules of Section~\ref{sec:calculus} in the first-order
theorem prover Vampire~\cite{Vampire13}. Note that Vampire behaves
well on theory problems with quantifiers both at the SMT and
first-order theorem proving competitions, winning respectively 5
divisions in the SMT-COMP 2016 competition of SMT
solvers\footnote{\url{http://smtcomp.sourceforge.net/2016/}} and the
quantified theory division of the CASC 2016 competition of first-order
provers~\footnote{\url{http://www.cs.miami.edu/~tptp/CASC/J8/}}.  With
our implementation, Vampire becomes the first superposition theorem
prover able to prove properties of term algebras. Moreover, our
experiments described later show that Vampire outperforms
state-of-the-art SMT solvers, such as CVC4 and Z3, on existing
benchmarks.

Our implementation required altogether about 2,500 lines of
C\texttt{++} code. The new version of Vampire, together with our
benchmark suite, is available for
download\footnote{\url{http://www.cse.chalmers.se/~simrob/tools.html}}.

\subsection{Input Syntax and Tool Usage}

In our work, we used an extended SMTLIB syntax~\cite{SMTLIB} to
describe term constructors. Although not yet part of the official
SMTLIB standard, this syntax is already supported by the SMT solvers
Z3 and CVC4, and its standardization is under consideration.

Our input syntax uses \texttt{declare-datatypes} for declaring an
abstract data type corresponding to a term algebra sort. This
declaration simultaneously adds the term algebra symbols and the
$\Sub$ predicate to the problem signature, adds the distinctness,
injectivity, domain closure and subterm axioms to the input set of
formulas, and activates the additional inferences rules from
Section~\ref{sec:calculus}.  Alternatively, the user can choose not to
activate the inference rules in our implementation. The inclusion of
the $\Sub$ predicate and its axioms, as presented in
Section~\ref{sec:extension}, can also be deactivated.

Note that the SMTLIB syntax also provides the not yet standardized
command \texttt{declare-codatatypes} to declare types of potentially
cyclic or infinite data structures. The theory underlying the
semantics of such types is almost identical to that of finite term
algebras, except that the acyclicity axiom is replaced by a uniqueness
rule that asserts that observationally equal terms are indeed
equal~\cite{reynolds2015codatatypes}. Therefore our calculus
\emph{minus the acyclicity axioms/rules} is an incomplete but sound
inference system for that theory, and users can declare co-algebraic
data types in their problems as well. Like acyclicity, the uniqueness
principle of co-algebras is not finitely axiomatizable.

\subsection{Benchmarks}\label{sec:benchmarks}

We evaluated our implementation on two sets of problems. These
problems included all publicly available benchmarks, as mentioned
below.
\begin{itemize}
\item A (parametrized) game theory problem originally described
  in~\cite{colmerauer2000expressiveness}. This problem relies on the
  term algebra of natural numbers to describe winning and losing
  positions of a game.  It is possible to encode, for a given positive
  integer $k$, a predicate $\mathit{winning}_k$ over positions, such
  that $\mathit{winning}_k(p)$ holds iff there exists a winning
  strategy from the position $p$ in $k$ or fewer moves. The
  satisfiability of the resulting first-order formula can be checked
  by term algebra decision procedures, since it does not use symbols
  other than those of the term algebra, but it includes $2k$
  alternating universal and existential quantifiers. This heavy use of
  quantifiers makes it an interesting and challenging problem for
  provers. An example of this problem encoded in the SMTLIB syntax is
  given in Figure~\ref{fig:smtlib}.

\item Problems about functional programs, generated by the Isabelle
  interactive theorem prover~\cite{Isabelle} and translated by the
  Sledgehammer system~\cite{blanchette2013sledgehammer}. The resulting
  SMTLIB problems include algebraic and co-algebraic data types as
  well as arbitrary types and function symbols, and also some
  quantified formulas. Some of these problems are taken from the
  Isabelle distribution (Distro) and the Archive of Formal Proofs
  (AFP), others from a theory about Bird and Stern–Brocot trees by
  Peter Gammie and Andreas Lochbihler (G\&L). They are representative
  of the kind of problems corresponding to program analysis and
  verification goals. This set of problems originally appeared
  in~\cite{reynolds2015codatatypes} and, to the best of our knowledge,
  represent the set of all publicly available benchmarks on algebraic
  data types.
\end{itemize}

\begin{figure}[tb]
\begin{verbatim}
(declare-datatypes ()
  ((Nat (z) (s (pred Nat)))))

(assert
  (not
    (exists
      ((w1 Nat)) 
      (and
        (or
          (= (s z) (s w1))
          (= (s z) (s (s w1)))
        )
        (forall
          ((l0 Nat))
          (=>
            (or
              (= w1 (s l0))
              (= w1 (s (s l0)))
            )
            false))))))

(check-sat)
\end{verbatim}
\caption{An instance of the game theory problem
  from~\cite{colmerauer2000expressiveness}, encoded in SMTLIB
  syntax. The first command declares a term algebra with a constant
  \texttt{z} and a unary function \texttt{s}; note that the projection
  function \texttt{pred} must also be named. The assertion (starting
  with \texttt{assert}) is a formula corresponding to the negation of
  the predicate $\mathit{winning}_1(s(z))$.}
\label{fig:smtlib}
\end{figure}

\subsection{Evaluation}\label{sec:eval}

Our experiments were carried out on a cluster on which each node is
equipped with two quad core Intel processors running at 2.4 GHz and
24GiB of memory. To compare our work to other state-of-the-art
systems, we include the results of running the SMT solvers Z3 and CVC4
on the Isabelle problems, as previously reported
in~\cite{reynolds2015codatatypes}, and also add the results of running
these two solvers on the game theory problem.

\noindent{\bf Game theory problems.} The times required to solve the
game theory problem for different values of the parameter $k$ are
shown in Table~\ref{tab:game}. The first column indicates the time
required by Vampire using the theory axioms (A) described in
Section~\ref{sec:extension}, and the second and third columns give the
time needed when the simplification rules (R) are also activated in
Vampire (Section~\ref{sec:calculus}). For this particular problem, the
acyclicity rule plays no role in the proof, but in order to assess its
impact on performance, the third column shows the times needed to
solve the problem when the subterm relation axioms (S) are also
included in the input. The fourth and fifth columns of
Table~\ref{tab:game} respectively indicate the times needed by CVC4
and Z3 for solving the corresponding problem. Where no value is given,
the prover was unable to solve the problem. Despite belonging to a
decidable class, this problem is quite challenging for theorem provers
and SMT solvers, which is easily explained by the presence of a
formula with many quantifier alternations. The SMT solver CVC4 is able
to disprove the negated conjecture only for $k = 1$, and Z3 can
disprove it only for $k = 1$ or $k = 2$. SMT solvers can also consider
the (non-negated) conjecture and try to satisfy it, but this does not
produce better results. In comparison, our implementation in Vampire
can solve the problem for $k = 6$, that is for formulas with 12
alternated existential and universal quantifiers, in 8.19
seconds. In~\cite{colmerauer2000expressiveness}, the authors are able
to solve the problem for $k$ as high as 80, using an implementation of
the decision procedure presented in~\cite{dao2000resolution}. However
such a decision procedure would not be able to reason in the presence
of uninterpreted symbols, and therefore its usage is much more
restricted. The results of Table~\ref{tab:game} confirm that
first-order provers can be better suited than SMT solvers for
reasoning about formulas with many quantifiers, despite the various
strategies used for quantifier reasoning in SMT solvers (for example,
by using E-matching~\cite{demoura2007ematching}). Table~\ref{tab:game}
also shows that adding simplification rules as described in
Section~\ref{sec:calculus} improves the behavior of the theorem
prover.

\begin{table}[tb]
  \centering
  \begin{tabular}{|c|c|c|c|c|c|}
    \hline
    $k$ & \begin{tabular}{@{}c@{}}Vampire \\ (A)\end{tabular} & \begin{tabular}{@{}c@{}}Vampire \\ (A+R)\end{tabular} & \begin{tabular}{@{}c@{}}Vampire \\ (A+R+S)\end{tabular} & CVC4 & Z3 \\
    \hline
    1 & 0.01 & 0.01 & 0.01 & 0.01 & 0.01 \\
    2 & 0.01 & 0.01 & 0.01 & 0.01 & 0.01 \\
    3 & 4.98 & 0.18 & 0.66 & -- & -- \\
    4 & 2.21 & 0.32 & 0.63 & -- & -- \\
    5 & 35.16 & 11.17 & 15.40 & -- & -- \\
    6 & 31.57 & 8.19 & 11.33 & -- & -- \\
    7 & -- & -- & -- & -- & -- \\
    \hline
  \end{tabular}
  \caption{Time required to prove unsatisfiability of different
    instances of the game theory problem
    from~\cite{colmerauer2000expressiveness}.}
  \label{tab:game}
\end{table}

\noindent{\bf Isabelle problems about functional programs.} Our
results on evaluating Vampire on the Isabelle problems are shown in
Table~\ref{tab:isabelle}. The problems were translated by Sledgehammer
by selecting some lemmas possibly relevant to a given proof goal in
Isabelle and translating them to SMTLIB along with the negation of the
goal. While the intent of this translation is to produce unsatisfiable
first-order problems, this is not the case for all of the problems
tested here. A few problems are satisfiable and it is likely that many
are unprovable, for example because the lemmas selected by
Sledgehammer are not sufficiently strong to prove the goal. The set of
problems originally included 4170 problems, of which 2869 include at
least one algebraic data type and 2825 include at least one
co-algebraic data type, some problems containing both. In the presence
of co-algebraic data types, CVC4 has a special decision procedure
which replaces the acyclicity rule by a uniqueness rule. In our
implementation, Vampire simply does not add the acyclicity axiom, but
the remaining axioms are added as they hold for co-algebraic data
types as well. Unlike CVC4, Z3 does not support reasoning about
co-algebraic data types.

In order to test the efficiency of our acyclicity techniques on more
examples, we considered problems containing co-algebraic data types:
by replacing them with algebraic data types with similar constructors,
we obtained different problems where the acyclicity principle
applies. Note that not all co-algebraic data type definitions
correspond to a well-founded definition for an algebraic data type:
after leaving these out, we obtained 2112 new problems.

Table~\ref{tab:isabelle} summarizes our results on this set of
benchmarks, using a single best strategy in Vampire. For each solver,
we also show the number of problems solved uniquely only by that
solver.

We also ran Vampire with a combination of strategies with a total time
limit of $120$ seconds.  Table~\ref{tab:distrib} shows the total
number of solved problems, with details on whether the problems
contain only algebraic data types, co-algebraic data types, or both.
Overall, Vampire is able to solve 1785 problems, that is 4,2\% more
that CVC4 and 7,3\% more than Z3, which is a significant
improvement. 50 problems are uniquely solved by Vampire, as listed in
column six Table~\ref{tab:distrib}. When compared to Vampire, only 4
problems were proved by CVC4 alone, while Z3 cannot prove any problem
that was not proved by Vampire -- see columns seven and eight of
Table~\ref{tab:distrib}. Summarizing, Table~\ref{tab:isabelle} shows
that Vampire outperforms the best existing solvers so far. The
experimental results of Tables~\ref{tab:game}-\ref{tab:isabelle}
provide an evidence that our methods for proving properties of
algebraic data types outperform methods currently used by SMT solvers.

\begin{table}[tb]
  \centering
  \begin{tabular}{|c|c|c|}
    \hline
    Prover & Solved & Unique \\
    \hline
    Z3 & 1665 & 5 \\
    CVC4 & 1711 & 12 \\
    Vampire (Best strategy) & 1720 & 31 \\
    \hline
  \end{tabular}
  \caption{Number of problems solved among the 6282 Isabelle problems
    translated by SledgeHammer.}
  \label{tab:isabelle}
\end{table}

\begin{table*}[tb]
  \centering
  \begin{tabular}{|c|c|c|c|c|c|c|c|}
    \hline
    & Total & Vampire & CVC4 & Z3 & Unique-Vampire & Unique-CVC4 & Unique-Z3 \\
    \hline
    Data types only & 3457 & 999 & 956 & 947 & 23 & 0 & 0 \\
    Co-data types only & 1301 & 430 & 415 & 382 & 16 & 2 & 0 \\
    Both & 1524 & 356 & 341 & 334 & 11 & 2 & 0 \\
    \hline
    Union & 6282 & 1785 & 1712 & 1663 & 50 & 4 & 0 \\
    \hline
  \end{tabular}
  \caption{Distribution of solved problems according to the data types
    they feature}
  \label{tab:distrib}
\end{table*}

\subsection{Comparison of Option Values}

We were also interested in comparing how various proof option values
affect the performance of a theorem prover. For the purpose of this
research, the options that we considered are:

\begin{enumerate}
\item the Boolean value selecting whether term algebra rules are used;
\item the value selecting how acyclicity is treated (axioms, rule, or
  none, that is, no acyclicity axioms or rules).
\end{enumerate}

Making such a comparison is hard, since there is no obvious
methodology for doing so, especially considering that Vampire has 64
options commonly used in experiments. The majority of these parameters
are boolean, some are finitely-valued, some integer-valued and some
range over other infinite domains. The method we used was based on the
following ideas. Suppose we want to compare values for an option
$\pi$. Then:

\begin{enumerate}
\item we use a set of problems obtained by discarding problems that
  are too easy or currently unsolvable;
\item we repeatedly select a random problem $P$ in this set, a random
  strategy $S$ and run $P$ on variants of $S$ obtained by choosing all
  possible values for $\pi$ using the same time limit.
\end{enumerate}
We discovered that the results for the term algebra rules are
inconclusive (turning them on or off makes little effect on the
results) and will present the results for the acyclicity option.

Our selected set of problems consisted of 262 term algebra
problems. We made 90,000 runs for each value (off, theory axioms, and
simplification rules), that is, 270,000 tests all together, with the
time limit of 30 seconds. While interpreting the results, it is worth
mentioning the following.

\begin{enumerate}
\item When neither acyclicity rules nor acyclicity axioms are used,
  problems whose proof rely on acyclicity become unsolvable. On the
  other hand, for other problems, this setting results in a smaller
  search space.

\item When acyclicity simplification rules are used, the resulting
  calculus is incomplete even for pure term algebra problems, but the
  subterm relation is not used, which generally means that fewer
  clauses should be generated.
\end{enumerate}

The results of these experiments are shown in
Table~\ref{tab:acyclicity}. We show the total number of successful
runs (out of 90,000) and the number of runs where only one value for
this option solved the problem. Probably the most interesting
observation is that using acyclicity simplification rules
(Section~\ref{sec:calculus}) instead of theory axioms
(Section~\ref{sec:extension}) results in many more problems solved.
This gives us an evidence that the axiomatization based on the subterm
relation results in much larger search spaces. This also means that
the value resulting in an incomplete strategy in this case generally
behaves better.

One should also note the 50 problems solved only when turning
acyclicity off.  This means that even the light-weight rule-based
treatment of acyclicity sometimes results in a large
overhead. Moreover, out of these 50 problems 10 were solved in less
than 1 second.

\begin{table}[tb]
  \centering
  \begin{tabular}{|l|c|c|c|}
    \hline
    & off & axioms & rules \\
    \hline
    Total solved & 2030 & 9086 & 9602 \\
    Solved by only this value & 50 & 70 & 566\\
    \hline
  \end{tabular}
  \caption{Comparison of proof option values for acyclicity in
    Vampire.}
  \label{tab:acyclicity}
\end{table}

\section{Related Work} 
\label{sec:related}
The problem of reasoning over term algebras first appears in the
restricted form of syntactic unification, mentioned
in~\cite{herbrand1930recherches}. The algorithm for syntactic
unification was later described in~\cite{Robinson65}, and later
refined into
quasi-linear~\cite{baxter1976complexity,huet1976resolution,martelli1982efficient}
and linear algorithms~\cite{paterson1976linear}.

The full-first order theory of term algebras over a finite signature
was first studied in~\cite{malcev1962axiomatizable}, where its
decidability was proved by quantifier elimination. Other quantifier
elimination procedures appeared in
\cite{comon1988unification,Maher88,Hodges:ModelTheory,RybinaVoronkov:termalgebras:2001}.
\cite{FerranteRackoff} proved a result implying that the first-order
theory of term algebras is non-elementary. There is a large body of
research on decidability of various extensions of term algebras, which
we do not describe here.

In this paper we do not prove decidability of new theories. However,
we present a new superposition-based decision procedure for
first-order theories of term algebras using a finitely axiomatizable
theory.

Probably the first implementation of a decision procedure for term
algebras is described in~\cite{colmerauer2000expressiveness}.  The
theory of finite or infinite trees is also studied
in~\cite{dao2000resolution} and a practical decision procedure is
given based on rewriting.

Due to recent applications of program analysis, there is now a growing
interest in the automated reasoning community for practical
implementation of term algebras and their combinations with other
theories. A decision procedure for algebraic data types is given in
~\cite{barrett2007decision} and later extended to a decision procedure
for co-algebraic data types in~\cite{reynolds2015codatatypes}. These
decision procedures exploit SMT-style reasoning and are supported by
CVC4. Z3 also supports proving properties about algebraic data
types~\cite{bjorner2013higher}.  Unlike these techniques, our work
targets the full first-order theory of term algebras, with arbitrary
use of quantifiers. Our proof search procedure is based on the
superposition calculus and allows one to prove properties with both
theories and quantifiers.

\section{Conclusion}
\label{sec:conclusion}
We presented two different ways to reason in the presence of the
theory of finite term algebras with a superposition-based first-order
theorem prover. Our first approach is based on a finitely
axiomatizable conservative extension of the theory and can be
implemented in any first-order theorem prover. The second technique
extends the first with the addition of extra inference and
simplification rules having two aims:

\begin{enumerate}
\item simplifying more clauses;
\item replacing expensive subterm-based reasoning about acyclicity by
  light-weight inference rules (though incomplete even without
  uninterpreted functions).
\end{enumerate}
While not as efficient as specialized decision procedures for this
theory, both our techniques allow us to reason about problems that
includes the theory of finite terms algebras and other predicate or
function symbols. We evaluated our work on game theory constraints and
properties of functional program manipulating algebraic data types.

The next natural development would be to extend our approach to the
theories of rational (finite but possibly cyclic) and infinite term
algebras. The notion of co-algebras is also closely related to
possibly infinite terms, with the addition of a uniqueness principle
for cyclic terms. A decision procedure for this theory was included in
the SMT solver CVC4 to decide problems involving co-algebraic data
types~\cite{reynolds2015codatatypes}. Co-algebras are also best suited
to express the semantics of processes and structures involving a
notion of state. Unlike term algebras, co-algebras have been studied
almost exclusively from the point of view of category theory, rather
than that of first-order logic, so that many theoretical and practical
applications remain to be explored there.

An even more interesting avenue to exploit is inductive reasoning
about algebraic data types in first-order theorem proving, also based
on extensions of the superposition calculus.

The work presented here should be a useful development for the
verification of functional programs. For example it would benefit the
tool HALO~\cite{vytiniotis2013halo}, which expresses the denotational
semantics of Haskell programs in first-order logic, before using
automated theorem provers to verify some of their properties. Our work
not only makes the translation easier but also modifies the prover to
make it more efficient on the generated problems. This also applies to
other tools that already use first-order theorem provers to discharge
their proof obligations, such as inductive theorem provers,
e.g.\ HipSpec~\cite{claessen2013automating} and automated reasoning
tools for higher-order logic,
e.g.\ Sledgehammer~\cite{blanchette2013sledgehammer}.

More generally, our work makes an important step towards closing the
gap between SMT solvers and first-order theorem provers. The former
are traditionally used for problems involving theories, while the
latter are better at dealing with quantifiers. Problems that include
both quantifiers and theories are very common in practical
applications and represent a big challenge due to their intrinsic
complexity, both in theory and in practice. Our results show that
first-order theorem provers can perform efficient reasoning in the
presence of theories, solving many problems previously unsolvable by
other tools.






\acks
We acknowledge funding from the ERC Starting Grant 2014 SYMCAR 639270, the Wallenberg Academy
Fellowship 2014 TheProSE, the Swedish VR grant GenPro D0497701, the
Austrian FWF 
research project RiSE S11409-N23,  and the EPSRC grant ReVeS: Reasoning
for Verification and Security.



\end{document}